\documentclass[11pt,oneside,letterpaper]{article}
\usepackage[mode=buildnew]{standalone}
\usepackage{eurosym}
\usepackage{graphicx,amsmath,amsfonts}
\usepackage{setspace}
\usepackage{amssymb}
\usepackage{dsfont}
\usepackage{amsmath}
\usepackage{amsfonts}
\usepackage{ifthen}
\usepackage{mathtools}
\usepackage{graphicx}
\usepackage{enumitem}
\usepackage{color}
\usepackage{framed}
\usepackage[margin=1in]{geometry}
\usepackage{natbib}
\usepackage[dvipsnames]{xcolor}
\usepackage[colorlinks=true, urlcolor=Mahogany, linkcolor=Mahogany, citecolor=Mahogany]{hyperref}
\usepackage{fp}
\usepackage{amsthm}
\usepackage{caption}
\usepackage{subcaption}
\usepackage{palatino}
\usepackage{cuted,balance}
\usepackage{amssymb}
\usepackage{soul}
\usepackage{graphicx,amsmath,amsfonts}
\usepackage{float}
\usepackage{bbm}
\usepackage{bm}
\usepackage{pgfplots}
\pgfplotsset{compat=1.18}
\usepackage{tikz}
\usetikzlibrary{positioning,intersections,pgfplots.fillbetween}

\usepackage{macros}

\usepackage[ruled]{algorithm2e}

\begin{document}

\title{Signaling in Data Markets via Free Samples}
\author{
Nivasini Ananthakrishnan\thanks{Department of Electrical Engineering and Computer Sciences, University of California, Berkeley} \and Alireza Fallah\thanks{Department of Computer Science and Ken Kennedy Institute, Rice University; Simons Institute for the Theory of Computing, University of California, Berkeley} \and Michael I. Jordan\thanks{Departments of Electrical Engineering and Computer Sciences and Statistics, University of California, Berkeley; Inria Paris}}

\date{\today}

\maketitle

\sloppy

\begin{abstract}
We study a setting in which a data buyer seeks to estimate an unknown parameter by purchasing samples from one of $K$ data sellers. Each seller has privately known data quality (e.g., high vs. low variance) and a private per-sample cost. We consider a multi-stage game in which the first stage is a free-trial stage in which the sellers have the option of signaling data quality by offering a few samples of data for free. Buyers update their beliefs based on the sample variance of the free data and then run a procurement auction to buy data in a second stage. For the auction stage, we characterize an approximately optimal Bayesian incentive compatible mechanism: the buyer selects a single seller by minimizing a belief-adjusted virtual cost and chooses the purchased sample size as a function of posterior quality and virtual cost. For the free-trial stage, we characterize the equilibrium, taking the above mechanism as the continuation game. Free trials may fail to emerge: for some parameters, all sellers reveal zero samples. However, under sufficiently strong competition (large $K$), there is an equilibrium in which sellers reveal the maximum allowable number of samples; in fact, it is the unique equilibrium.
\end{abstract}

\section{Introduction}
Data markets are increasingly central to economic and scientific activity: a decision maker can purchase datasets from competing providers to estimate a latent quantity of interest, yet the quality of these datasets is often only known to the sellers and is difficult to verify ex ante. A natural response---mirroring ``free trials'' in software and consumer goods---is to let sellers provide a small number of samples for free so that buyers can assess quality before purchasing. Whether such free trials arise in equilibrium, however, is not obvious. Data sellers may not have full control over the quality of their raw data, and it may be costly to perform quality control, so that committing in advance to providing free samples can be risky. Moreover, free trials can intensify competition and reduce informational rents. This paper studies when, and why, free trials emerge in competitive data markets, and when they fail to do so.

We model a market with $K$ data sellers and a continuum of buyers. A buyer seeks to estimate an unknown parameter $\theta$ and can obtain i.i.d.\ samples from seller $i$ drawn from a normal distribution $N(\theta,\sigma_i^2)$, where $\sigma_i \in \{\sigma_L,\sigma_H\}$ captures the seller's (privately known) data quality, with $\sigma_L < \sigma_H$. Each seller also has a private per-sample cost $c_i$, drawn independently from a continuous distribution. The interaction is a multi-stage game. First, before any buyer arrives, each seller chooses a free-trial size $m_i$, between zero and a maximum $M$, interpreted as the number of samples it will reveal for free. Next, a buyer arrives; each seller observes the sample quality it can offer for that buyer’s demand and the cost at which it can do so. Then, each seller submits a cost report $c_i'$ and a free-sample dataset of size $m_i$. 
We assume that this dataset is used only for quality evaluation and cannot be used in the final estimation, which is consistent with “evaluation-only” licensing and other practical restrictions, as discussed in the model section.
The buyer then forms beliefs about each seller's quality from the empirical variance of the submitted dataset. Given these beliefs and the cost reports, the buyer runs a procurement mechanism to decide from whom to purchase additional samples, how many to purchase, and what to pay. The buyer then estimates $\theta$ using a weighted average of the purchased samples. 

A key challenge is that the free-trial stage feeds into a downstream Bayesian mechanism-design problem in which the allocation must be Bayesian incentive compatible (BIC) in sellers’ reported costs, while the buyer’s objective trades off estimation precision and payments. We first derive an approximately optimal mechanism for the procurement stage and then analyze equilibrium free-trial choices, taking this mechanism as the continuation game. Concretely, we relax the buyer’s problem by allowing the purchased sample size to be real-valued, solve the relaxed problem, and then obtain a feasible mechanism by rounding down the resulting sample size. We show that, under standard regularity conditions on the cost distribution, the optimal mechanism for the relaxed problem (and hence its rounded version) has a simple \textit{``single-sourcing" }form: the buyer purchases data from only one seller.

In particular, using the free samples, the buyer forms posterior beliefs and then, for each seller, computes a belief-adjusted cost index that combines the seller’s reported (virtual) cost with the buyer’s posterior assessment of that seller’s variance. The mechanism selects a single winning seller $i^\star$ by optimizing this index—equivalently, by choosing the seller with the best trade-off between low inferred variance and low virtual cost. By rounding down the resulting allocation and choosing the payment rule via a standard Myerson-style construction that ensures BIC, we obtain a mechanism that is approximately optimal in terms of buyer utility: it incurs only a small loss relative to the optimal BIC mechanism, while preserving the clean structural characterization of single-sourcing mechanisms.

Motivated by this result, we analyze the induced free-trial stage when the buyer uses the mechanism above as the continuation game. In particular, we characterize the approximate subgame-perfect Bayesian Nash equilibria in sellers' choices of $(m_i)_{i=1}^K$ (approximation enters only through the buyer’s continuation policy; sellers best-respond exactly). Our first theorem shows that free trials can fail to emerge even in competitive markets.

\paragraph{Theorem 1 (Informal)}
For any number of sellers $K$, there exist range of parameters, i.e., choices of $(\sigma_L,\sigma_H)$ and the prior belief of the buyer, in which there exists an approximate equilibrium where all sellers reveal no free samples, $m_i = 0$ for every $i$.

\vspace{2mm}

The main intuition behind this ``uninformative'' equilibrium is that committing to share free samples carries two risks. First, even when a seller is truly low-variance, although revealing samples makes the buyer more confident in selecting them, it can also reduce the seller’s payment conditional on being selected, because the buyer believes fewer samples are needed to achieve the desired accuracy. In other words, the seller may win the competition but end up selling fewer samples than they would have sold had they won under the buyer’s prior belief without sharing any samples. Second, if the seller’s realized data quality is poor (i.e., high-variance), revealing samples makes the buyer more pessimistic and can sharply reduce the probability of being selected. In some parameter regimes, these two effects dominate any benefit from signaling, so the safe strategy is to reveal nothing, and no seller wants to deviate unilaterally from offering zero free samples. 

Our second theorem shows that the picture reverses under sufficiently strong competition.

\paragraph{Theorem 2 (Informal)}
For any fixed underlying parameters, there exists a threshold $\bar K$ such that for all $K \ge \bar K$, there is an approximate equilibrium in which every seller reveals the maximum allowable number of free samples, $m_i = M$ for all $i$. Moreover, assuming the buyer utilizes the approximate mechanism presented above, this is the unique equilibrium.

\vspace{2mm}

The central insight here is that with many sellers, selection becomes extremely sensitive to small differences in the buyer's posterior assessment of quality. Providing more free samples tightens the posterior, pushing the seller's posterior mean variance closer to $\sigma_L$ in favorable realizations; because the procurement mechanism selects the seller with the best belief-adjusted virtual cost index, even modest improvements in perceived quality can translate into a first-order advantage in winning probability. When $K$ is large, this competitive pressure overwhelms the cost of giving away samples: deviating to fewer than $M$ samples makes it increasingly likely that some rival both looks sufficiently high-quality and has sufficiently low cost, causing the deviator's selection probability (and hence expected profit) to collapse. As a result, maximal disclosure becomes self-enforcing, and in fact becomes the only equilibrium as $K$ grows. 

Finally, to explore the landscape of symmetric equilibria beyond the two extremes identified in Theorems 1 and 2, we conduct numerical simulations across a range of problem parameters $(\sigma_L,\sigma_H,K,M)$. We find that there exist symmetric equilibria in which every seller chooses an intermediate free-trial size, $m_i=m^*$ with $0 < m^* < M$, and that multiple symmetric equilibria, corresponding to different disclosure levels, can coexist for the same parameter regime.

Our results highlight how ``free samples'' can serve as a signal of data quality in data markets. Free trials are not guaranteed to arise simply because they can be informative; in some parameter regimes, they unravel entirely and the market can remain maximally opaque even with many sellers. At the same time, once competition is sufficiently intense, free trials become a powerful competitive instrument: sellers disclose aggressively to avoid being screened out by the buyer’s belief-adjusted selection rule, and the market converges to full disclosure.

\subsection{Related Work}\label{sec:related_work}

\paragraph{Data markets.} Our work lies in the literature on data markets~\citep[see, e.g.,][]{bergemann2019markets,acemoglu2022too}, and more specifically in the design of mechanisms for data acquisition to estimate an underlying population statistic~\citep{cummings2015accuracy,chen2018optimal, fallah2024optimal, fallah2022bridging, cummings2023optimal, liao2024privacy}, which is closest to our setting, as well as mechanisms for learning machine learning models~\citep{agarwal2019marketplace, cai2015optimum, li2021data, chen2022selling, saig2023delegated,ananthakrishnan2024delegating}. In many of these works, the data buyer is assumed to know, observe, or control data quality through mechanism design or the design of sampling and evaluation strategies. In contrast, our work allows sellers to opt in to sharing samples and thereby influence the resulting equilibrium behavior.

\paragraph{Strategic information revelation in data markets.} Strategic information revelation in data markets has been studied for various modes of information revelation. For example, sellers may disclose reviews or ratings~\citep{ifrach2019bayesian, boursier2022social,han2020customer, chamley2004rational,
crapis2017monopoly, besbes2018information, acemoglu2022learning}. Data owners may obfuscate the quality of their data to guard against price discrimination~\citep{acquisti2005conditioning, taylor2004consumer, conitzer2012hide, ichihashi2020online}. The mode we study in our work is that of revealing information through providing free samples. 

\paragraph{Information revelation through free samples.} Previous work by~\cite{drakopoulos2023providing} also studies free samples in data markets. They focus on the \emph{hold up} problem, which involves preventing the buyer from free-riding on the seller's free samples. In contrast, our focus is on how free samples shape buyer's beliefs about the quality of data. There is a rich literature on free samples or trials in broader settings~\citep{nelson1970information, shapiro1983premiums, bergemann2006dynamic, heiman2001learning, bawa2004effects, boleslavsky2017demonstrations, cheng2010free, board2018competitive}.

\paragraph{Competition and externalities in data markets.}
Previous work has studied how data platforms shape competition among sellers~\citep{bergemann2015selling, bergemann2022economics, ichihashi2021economics, de2025data, fallah2024three}. There is also work studying platform level competition and coordination in data markets~\citep{ichihashi2021competing, gu2022data, abrardi2025data}. Finally, there is work studying competition among buyers in data markets~\citep{bonatti2024selling, montes2019value, bimpikis2019information}. Our work focuses on the effect of competition in data markets when there is an option to provide free samples.

\paragraph{Persuasion and information design.} The role of free samples in shaping buyer's beliefs about data quality is closely related to persuasion and information design~\citep{kamenica2011bayesian, bergemann2019information,kamenica2019bayesian}. Particularly related is work on how competition increases information revelation~\citep{gentzkow2016competition, au2020competitive}. Our results can be viewed as an analogue where there are structural constraints on how signals can be generated; namely, through providing free samples, only the sample size can be controlled. Persuasion and information design has been studied in the context of auctions~\citep{agarwal2024towards, bergemann2017first, esHo2007optimal, bergemann2007information}, pricing~\citep{bergemann2018design}, and mechanism design~\citep{babaioff2012optimal,bonatti2024selling}.
\section{Model}

We consider a market with $K$ data sellers who interact with a continuum (mass) of data buyers, denoted by $\mathcal{B}$. Each buyer $B \in \mathcal{B}$ seeks to estimate a parameter $\theta(B)$. For a given buyer $B$, seller $i \in [K]$ can provide i.i.d.\ samples from a normal distribution, $\mathcal{N}(\theta(B), \sigma_i(B)^2)$, where $\sigma_i(B) \in \{\sigma_L, \sigma_H\}$ represents the quality of seller $i$'s data for buyer $B$, with $\sigma_L < \sigma_H$. The cost per sample for seller $i$ to provide data to buyer $B$ is denoted by $c_i(B)$. This cost can be interpreted either as the operational cost of generating/collecting data (e.g., running a test or using an imaging device) or as a privacy cost. When the buyer $B$ is clear from the context, we suppress the dependence of $\theta(\cdot)$, $\sigma_i(\cdot)$, and $c_i(\cdot)$ on $B$.

Ex ante (before the buyer is realized), the data variance for each seller is distributed as a Bernoulli random variable that takes value $\sigma_L$ with probability $\mu$ and $\sigma_H$ with probability $1-\mu$. The per-sample cost is drawn (ex ante) from a continuous distribution with pdf $f$ and cdf $F$ supported on $[c_{\min}, c_{\max}]$, where $c_{\min} > 0$. We assume these ex-ante distributions are independent and identical across sellers, although for a given buyer $B$ the realizations $(\sigma_i(B), c_i(B))$ may differ across sellers.

\paragraph{Timing and information.}
We model the interaction between sellers and a buyer as a multi-stage game. First, upon entering the market, each seller $i$ chooses the number of free samples, $m_i \in {0,\dots,M}$, they will share as a signal of data quality. Here, $m_i$ is a commitment policy: sellers do not yet share any samples, because the realized sample quality is not yet known.

Next, a buyer $B \in \mathcal{B}$ arrives, at which point the sample variance $\sigma_i^2$ and per-sample cost $c_i$ are realized for each seller $i$. The buyer then runs a procurement auction: each seller $i$ submits (i) a reported per-sample cost $c_i’$ and (ii) a dataset of free samples $\mathcal{S}_i^f$ of size $m_i$. Sellers need not report costs truthfully, and neither sellers’ true costs nor their true qualities are observed by the buyer. However, we assume the shared free samples are i.i.d. draws from the same distribution as the final (potentially purchased) data, so sellers cannot manipulate or cherry-pick their free samples.

Based on the free samples, the buyer forms a belief about seller $i$’s variance, denoted by $\pi_i(\mathcal{S}_i^f)$ (or simply $\pi_i$ when clear from context). We assume this belief depends on the sample variance of the reported free samples in $\mathcal{S}_i^f$. This restriction is imposed because sample variance is a mean-invariant statistic, which allows the buyer to update beliefs about data quality without specifying a prior over the mean $\theta(B)$. Formally, $\pi_i(\mathcal{S}^f_i)$ is obtained by Bayes’ rule from the Bernoulli prior that assigns probability $\mu$ to $\sigma_i^2=\sigma_L^2$ and $1-\mu$ to $\sigma_i^2=\sigma_H^2$, using the observed sample variance of the free samples in $\mathcal{S}^f_i$ as the signal (i.e., the likelihood is the sampling distribution of the sample variance under $\mathcal{N}(\theta(B),\sigma_i^2)$). Lastly, to ensure that the sample variance is well defined, we assume that $m_i$ is never equal to $1$; that is, $m_i \in \{0, 2, \cdots, M\}$.

Given reported costs $\bm{c}'$ and beliefs $\bm{\pi}$, the buyer chooses how many samples to purchase from each seller and the corresponding payments. Specifically, the buyer purchases a dataset $\mathcal{S}_i := \{x_1^i, \ldots, x_{n_i}^i\}$ of size $n_i(\bm{c}'; \bm{\pi})$ from seller $i$ and pays $t_i(\bm{c}'; \bm{\pi})$.

It is important to emphasize that the free samples can be used only for evaluation (belief updating) and are not retained as inputs for estimation. This assumption is motivated by the fact that, in practice, trial access to data is often governed by evaluation-only licensing and/or delivered through controlled environments that permit inference or measurement but restrict retention and downstream reuse, which may include model training  \citep[see, e.g.,][]{PreciselyDataEULA2021}.

\paragraph{Estimation.}
The buyer uses the purchased data to estimate $\theta(B)$ via the weighted average
\begin{equation}
\hat{\theta}(\mathcal{S}, \bm{w}) := \sum_{i=1}^K w_i \frac{1}{n_i} \sum_{j=1}^{n_i} x_j^i,  
\end{equation}
where $w_i \ge 0$ for all $i$ and $\sum_{i=1}^K w_i = 1$. Under the buyer's beliefs, the variance of this estimator is
\begin{equation}
\text{Var}(\bm{n}, \bm{w}; \bm{\pi})  = \sum_{i=1}^K w_i^2 \frac{\bar{\sigma}_i^2}{n_i},     
\end{equation}
where $\bar{\sigma}_i^2$ denotes the mean of seller $i$'s belief distribution over variances:
\begin{equation}
\bar{\sigma}_i^2 = \pi_i(\sigma_L^2) \sigma_L^2 + \pi_i(\sigma_H^2) \sigma_H^2.    
\end{equation}

\paragraph{Utility functions.}
Seller $i$'s utility equals the payment received minus the cost of producing the samples sold:
\begin{equation}
U_i(c_i, \bm{c}'; \bm{\pi}; \bm{n}, \bm{t}) := t_i(\bm{c}'; \bm{\pi}) - n_i(\bm{c}'; \bm{\pi}) c_i.   
\end{equation}
The buyer's utility trades off estimation precision against total payments:
\begin{align}
U_b(\bm{c}', \bm{w}; \bm{\pi}; \bm{n}, \bm{t}) &:= -\text{Var}(\bm{n}, \bm{w}; \bm{\pi}) - \lambda \sum_{i=1}^K t_i(\bm{c}'; \bm{\pi}),   
\end{align}
where $\lambda$ captures the buyer's relative weight on payments versus estimation error.

\paragraph{Solution concept.}
We use (approximate) subgame-perfect Bayesian Nash equilibrium. Fix any belief profile $\bm{\pi}$. By the revelation principle, we can restrict attention to direct mechanisms, so the allocation and payment rules $(\bm{n}, \bm{t})$ solve
\begin{subequations} \label{eqn:mechanism_design}
\begin{align}
\max_{\bm{w}, \bm{t}, \bm{n}}\, & \mathbb{E}_{\bm{c}} \left[ U_b(\bm{c}, \bm{w}; \bm{\pi}; \bm{n}, \bm{t}) \right] \label{eqn:mechanism_design_obj} \\
\quad \text{s.t.} & \quad  \mathbb{E}_{\bm{c}_{-i}} \left[ U_i(c_i, c_i, \bm{c}_{-i}; \bm{\pi}; \bm{n}, \bm{t}) \right] \geq \mathbb{E}_{\bm{c}_{-i}} \left[ U_i(c_i, c'_i, \bm{c}_{-i}; \bm{\pi}; \bm{n}, \bm{t}) \right] 
\quad \text{for all } i, c_i, c'_i,
\label{eqn:mechanism_design_IC} \\
& \quad  \mathbb{E}_{\bm{c}_{-i}} \left[ U_i(c_i, c_i, \bm{c}_{-i}; \bm{\pi}; \bm{n}, \bm{t}) \right] \geq 0 
\quad \text{for all } i, c_i.
\label{eqn:mechanism_design_IR}
\end{align}
\end{subequations}
Constraint \eqref{eqn:mechanism_design_IC} is Bayesian incentive compatibility (BIC): seller $i$ prefers truthful cost reporting when other sellers report truthfully. Constraint \eqref{eqn:mechanism_design_IR} is interim individual rationality (IR): each seller's expected utility from participating is nonnegative. 

We further say $(\bm{n}^*, \bm{t}^*)$ is a $\zeta$-approximate solution to the mechanism design problem \eqref{eqn:mechanism_design} if it satisfies the IC and IR constraints \eqref{eqn:mechanism_design_IC} and \eqref{eqn:mechanism_design_IR} exactly, and achieves expected buyer cost (negative of buyer utility) at most a factor $1+\zeta$ of the optimal expected buyer cost.

Next, fixing a buyer continuation mechanism $(\bm{n}^*, \bm{t}^*)$, we define seller $i$'s interim (ex-ante) payoff under a free-sample profile $\bm{m}$ by
\begin{equation}
\bar{U}_i(\bm{m}) := \mathbb{E}_{\bm{c}, \bm{\mathcal{S}}^f} \left [ U_i\!\left(c_i, \bm{c}; \{\pi_i(\mathcal{S}_i^f)\}_{i=1}^K; \bm{n}^*, \bm{t}^*\right)  \right],     
\end{equation}
where $|\mathcal{S}_i^f| = m_i$ for each $i$. This is seller $i$'s expected utility at the time sellers choose free-sample commitments, when sellers provide free samples according to $\bm{m}$ and the subsequent procurement stage uses $(\bm{n}^*, \bm{t}^*)$.

Lastly, we say $(\bm{m}^*, \bm{n}^*, \bm{t}^*)$ is a $\zeta$-approximate subgame perfect equilibrium if (i) $(\bm{n}^*, \bm{t}^*)$ is an $\zeta$-approximate solution to \eqref{eqn:mechanism_design}, and (ii) for every seller $i$,
\begin{equation}
\bar{U}_i(\bm{m}^*) \ge \bar{U}_i(m_i', \bm{m}_{-i}^*)
\end{equation}
for all $m_i' \in \{0,1,\ldots,M\}$. It is worth noting that the approximation enters only through the buyer's continuation mechanism: sellers' incentives are evaluated exactly, and sellers have no profitable deviation in the free-sample stage.
\section{The Mechanism Design Subproblem}

We start by analyzing the buyer’s choice of allocation and payment functions, given the free samples received from all data sellers and the induced belief profile $\bm{\pi}$. We first make the following customary regularity assumption on the distribution of sample costs, i.e., $F(\cdot)$, which ensures that the \textit{virtual cost}, as defined below, is non-decreasing.
\begin{assumption}[Monotonicity of virtual costs]\label{ass:virtual_costs_monotonicity}
    The virtual cost $\psi(c):= c + F(c)/f(c)$ is a non-decreasing function of $c$.
\end{assumption}
This assumption is standard in mechanism design and holds for a variety of distributions including uniform, normal, and exponential distributions  \citep[see, e.g.,][]{rosling2002inventory}.

To solve the mechanism design problem \eqref{eqn:mechanism_design}, we first relax the objective function in \Cref{eqn:mechanism_design_obj} by allowing the allocation of the sample size purchased from each seller to be a real-valued number instead of an integer. We then round down this real-valued number to the nearest lower integer to obtain the sample allocation, and we show that the resulting allocation is approximately optimal.

\begin{proposition}[Approximately Optimal Purchasing and Payment Rules]\label{prop:mechanism_soln}
Suppose \Cref{ass:virtual_costs_monotonicity} holds.
Given the sequence of free samples shared by the sellers $\left (\mathcal{S}^f_i \right )_{i=1}^K$ inducing beliefs $({\pi})_{i=1}^K$ with means $(\bar{\sigma}_i^2)_{i=1}^K$, there is a sample purchasing and payment rule that forms a Bayesian Incentive Compatible mechanism with the following properties. When the sellers report costs $c^{'}_1, \ldots, c^{'}_K$,
\begin{enumerate}
    \item The buyer buys only from one seller $i^* \in [K]$ with $i^* \in \text{argmin}_{i \in [K]} \bar{\sigma}_i^2 \psi({c'_i})$.

    \item The buyer purchases $\lfloor n^* \rfloor$ samples from seller $i^*$, where $n^* = \bar{\sigma_{i^*}} / \sqrt{\lambda \psi(c^{'}_{i^*})}$.

    \item The expected payment each seller receives is the expectation of product of the number of samples purchased from the seller and the virtual cost of the seller $\mathbb{E}_{\bm{c}}[n_i(\bm{c}; \bm{\pi}) \psi(c_i)]$. 

    \item This mechanism yields a utility that is suboptimal by a factor of $1 - \frac{1}{n^* - 1}$ compared to the optimal BIC mechanism.  
    \end{enumerate}
     
\end{proposition}

The remainder of this section is devoted to proving this result. Before doing so, however, we make two remarks. First, as the proof will establish, the optimal mechanism corresponding to the aforementioned relaxed problem is \textit{single-sourced}, meaning that the buyer purchases all samples (i.e., $n^*$ samples) from a single data seller, denoted by seller $i^*$. Our rounded allocation preserves this property as well.

Second, it is straightforward to see that, under the assumption 
\begin{equation} \label{eqn:approx_condition}
\sigma_L \geq \alpha \sqrt{\lambda \psi(c_{\max})},    
\end{equation}
for some positive constant $\alpha \geq 3$, we can ensure that $n^* \geq \alpha$, which in turn guarantees a $(\alpha-2)/(\alpha-1)$-approximation factor. If we are satisfied with high-probability lower bounds on $n^*$, it is straightforward to see that $\psi(c_{\max})$ can be replaced in this condition by a smaller constant, since the winning data seller $i^*$ is likely to have a low cost. While we do not formalize this point here, we emphasize that assumptions of this form are generally unavoidable; otherwise, if $\lambda$ or the costs are too high, the cost of purchasing samples may induce the data buyer to purchase no samples at all. We next present the proof.

\begin{proof}[Proof of \Cref{prop:mechanism_soln}]

    Let us first find the optimal choice of weights $\bm{w}$ to produce the least variance estimator given a choice of $\bm{n}$ and $\bm{t}$. This will allow us to rewrite the mechanism design problem in terms of optimizing over $\bm{n}$ and $\bm{t}$. 

    \begin{lemma}[Optimal weights to minimize variance]\label{lem:opt_weights}
        Given the sequence of samples, purchases, payments, and means of posteriors $\bm{n}, \bm{t}, \bar{\bm{{\sigma}}}$, the vector of weights maximizing buyer utility $U_b$ is $\bm{w}^*$ with
        \[
            w_i^* = \frac{n_i/\bar{\sigma_i}^2}{\sum_{i\in[K]} n_i/\bar{\sigma_i}^2}.
        \]
    \end{lemma}
    \begin{proof}[Proof of \Cref{lem:opt_weights}]
        Utility function $U_b$ depends on $\bm{w}$ only through $\text{Var}(\bm{n}, \bm{w}; \bm{\pi})  = \sum_{i=1}^K w_i^2 \frac{\bar{\sigma}_i^2}{n_i}$. Hence the optimal choice of $\bm{w}$, given $\bm{n}, \bm{t}, \bar{\bm{{\sigma}}}$ minimizes $\sum_{i=1}^K w_i^2 \frac{\bar{\sigma}_i^2}{n_i}$ subject to the constraints that $\bm{w} \ge 0$ and $\sum_{i=1}^K w_i = 1$.

        By the Cauchy-Schwarz inequality, 
        \[1 = \left ( \sum_{i=1}^K w_i \right )^2 = \left (\sum_{i=1}^K \frac {\bar{\sigma}_i} {\sqrt{n_i}} w_i \cdot \frac  {\sqrt{n_i}} {\bar{\sigma}_i} \right )^2 \le \left ( \sum_{i=1}^K \frac {\bar{\sigma}_i^2} {{n_i}} w_i^2 \right ) \left ( \sum_{i=1}^K \frac  {{n_i}} {\bar{\sigma}_i^2} \right ) \]
        \[\Longrightarrow \sum_{i=1}^K \frac {\bar{\sigma}_i^2} {{n_i}} w_i^2 \ge \frac 1 {\sum_{i=1}^K \frac  {{n_i}} {\bar{\sigma}_i^2}}.\]

        The ${\bm{w}^*}$ with $w_i^* = \frac{n_i/\bar{\sigma_i}^2}{\sum_{i\in[K]} n_i/\bar{\sigma_i}^2}$  satisfies this lower bound with equality and hence is optimal.
    \end{proof}

Each seller's utility, Bayesian incentive compatibility (IC) and individual rationality (IR) have the same form as in a standard procurement auction. The utility of seller $i$ with cost $c_i$ when reported costs are $\Tilde{\bm{c}}$ and purchasing sample size and payment rules are $\bm{n}$ and $\bm{t}$ is $U_i(\Tilde{\bm{c}}; \bm{n}, \bm{t}; c_i) = t_i(\Tilde{\bm{c}}) - c_i \bm{n}_i(\Tilde{\bm{c}})$ . The Bayesian IC and IR constraints are in terms of the expected utility $\bar{U}_i(c_i, c^{'}_i, \bm{n}, \bm{t}) = \mathbb{E}_{\bm{c}_{-i}}[U_i(c_i, c^{'}_i, \bm{c}_{-i}; \bm{n}((c^{'}_i, \bm{c}_{-i})), \bm{t}((c^{'}_i, \bm{c}_{-i})))]$, which is the expected utility when seller $i$'s realized cost is $c_i$, taking an expectation over other sellers' costs, when seller $i$ reports cost $c^{'}_i$ and all other sellers truthfully report costs. Let us call this the \emph{cost-interim utility} of the seller.   
    We can write the cost-interim utility as $\bar{t}(c^{'}_i, c_i) - c_i \bar{n}(c^{'}_i, c_i)$ where $\bar{n}_i(c^{'}_i) = \mathbb{E}_{\bm{c}_{-i}}[n_i(c^{'}_i, \bm{c}_{-i})]$ and $\bar{t}_i(c^{'}_i) = \mathbb{E}_{\bm{c}_{-i}}[t_i(c^{'}_i, \bm{c}_{-i})]$ are the expected number of samples and payment received respectively by seller $i$ when reporting cost $c^{'}_i$. Again, the expectation is over the costs of the other sellers. Let us call $\bar{n}_i(c^{'}_i)$ the \emph{cost-interim sample-purchasing rule} and $\bar{t}_i(c^{'}_i)$ the \emph{cost-interim payment rule}.

    By Myerson's lemma~\citep{myerson1981optimal}, a payment and sample purchasing rule is Bayesian Incentive Compatible (BIC) if and only if the following properties are satisfied:

    \begin{enumerate}
        \item Monotonicity: The cost-interim sample purchasing rule is non-increasing in ${c_i}'$. That is, more expected samples must not be purchased if a seller reports a higher cost. 

        \item Myerson payment: The cost-interim payment satisfies 
        \[\bar{t}_i(c^{'}_i) = c^{'}_i \bar{n}_i(c^{'}_i) + U_i(c_{\text{max}}) + \int_{c^{'}_i}^{c_{\text{max}}} \bar{n}_i(t) dt.\]

        The cost-interim IR constraint can be written as $U_i(c_{\text{max}}) \ge 0$. To design the BIC mechanism with the least payment made, we can set $U_i(c_{\text{max}}) = 0$.
    \end{enumerate}
 As a consequence of the Myerson payment rule, we can write the expectation over each seller's cost $c_i$ of the cost-interim payment in the following way. Note that this expectation is the interim payment---which is the expected payment the buyer provides to a seller, taking an expectation over the realization of all sellers' costs. We have:
\begin{align*}
     \mathbb{E}_{\bm{c}} [t_i(\bm{c})] &= \mathbb{E}_{c_i} [\bar{t}_i(c_i)] \\
     &= \int_{c_\text{min}}^{{c_\text{max}}} c \bar{n}_i(c) f(c) dc + \int_{c_\text{min}}^{{c_\text{max}}} \left ( \int_{c}^{{c_\text{max}}} \bar{n}_i(t) dt \right ) f(c) dc \\
     &= \int_{c_\text{min}}^{{c_\text{max}}} c \bar{n}_i(c) f(c) dc +   \int_{c_\text{min}}^{{c_\text{max}}} \bar{n}_i(t) \left(\int_{{c_\text{min}}}^{c} f(c) dc \right ) dt \\
     &= \int_{c_\text{min}}^{{c_\text{max}}} \bar{n}_i(c) \left ( c + \frac {F(c)}{f(c)}\right ) f(c) dc \\
     &= \mathbb{E}_{c_i} [\bar{n}_i(c_i) \psi(c_i)].
 \end{align*}
Hence, the expected payment made in the optimal BIC mechanism is the product of the expected cost-interim purchased sample size and the virtual cost.  

With the Myerson payment, we can write the expected payment in the objective of the buyer's mechanism design problem $\mathbb{E}_{\bm{c}}[\lambda \sum_{i=1}^K t_i]$ as $\mathbb{E}_{\bm{c}}[\lambda \sum_{i=1}^K \bar{n}_i(c_i) \psi(c_i)]$, which is equal to $\mathbb{E}_{\bm{c}}[\lambda \sum_{i=1}^K n_i(\bm{c}) \psi(c_i)]$. We can then write the buyer's optimal BIC mechanism problem as the following optimization problem over sample-purchasing rules subject to the monotonicity constraint~\eqref{eq:monotonicity_myerson}:
\begin{align}
    \min_{\bm{n}} \quad & \mathbb{E}_{\bm{c}} \left [\frac{1}{\sum_{i=1}^K n_i(\bm{c}) / \bar{\sigma}_i^2} + \lambda \sum_{i=1}^K n_i(\bm{c}) \psi(c_i) \right ] \label{eq:sample_purchasing_problem}
    \\
    \text{subject to} \quad & \mathbb{E}_{\bm{c}_{-i}}[{\bm{n}}_i(\bm{c}_{-i},c_i)] \ge \mathbb{E}_{\bm{c}_{-i}}[{\bm{n}}_i(\bm{c}_{-i},c_i')]
    \quad \text{for all } i \in [K],\ c_i \le c_i'
    \label{eq:monotonicity_myerson}
\end{align}

Let us first minimize the objective~\eqref{eq:sample_purchasing_problem}, dropping the monotonicity constraint in~\eqref{eq:monotonicity_myerson}. We will later show that the solution satisfies the monotonicity constraint under \Cref{ass:virtual_costs_monotonicity}. Let us minimize the objective in \eqref{eq:sample_purchasing_problem}) by minimizing it for every realized $\bm{c}$. Let $\sum_{i=1}^K n_i \psi(c_i) = B$. Then, 
\[\sum_{i=1}^K n_i / \bar{\sigma}_i^2 = \sum_{i=1}^K \frac 1 {\bar{\sigma}_i^2 \psi(c_i)}n_i \psi(c_i) \le \max_{i \in [K]} \frac B {\bar{\sigma}_i^2 \psi(c_i)}.\]

For any fixed value of payments, the smallest value of the variance is achieved by only purchasing samples from sellers with the least $\bar{\sigma}_i^2 \psi(c_i)$. Let $i^*$ be such a seller. So the optimization problem becomes optimizing over the value of payments $B$: 

\[\min_{B \ge 0} \frac {\bar{\sigma}_{i^*}^2 \psi(c_{i^*})} B + \lambda B.\]

The optimal value is $B = \bar{\sigma}_{i^*} \sqrt{\psi(c_{i^*})/\lambda}$. We will consider the optimal mechanism that always chooses to purchase from one seller, breaking any ties by selecting uniformly at random. The number of samples purchased from the selected seller $i^*$ is then $n_{i^*}$, chosen so that $n_{i^*} \psi(c_{i^*}) = \bar{\sigma}_{i^*} \sqrt{\psi(c_{i^*})/\lambda}$. In other words, $n_{i^*} = \bar{\sigma}_i / \sqrt{\lambda \psi(c_{i^*})}$.

Now let us show that the minimizer without constraint~\eqref{eq:monotonicity_myerson} satisfies the constraint nonetheless. We will argue that the optimal sample purchasing rule is monotone in the reported cost as long as \Cref{ass:virtual_costs_monotonicity} holds (virtual costs are non-decreasing functions of costs). This implies the monotonicity constraint \eqref{eq:monotonicity_myerson} is satisfied.  Reporting a higher cost cannot increase chance of selection since the seller with the least $\bar{\sigma}_i^2 \psi(c_i)$ is selected. Additionally, the number of samples purchased from the selected seller is a non-increasing function of the virtual cost of the reported cost, and hence of the reported cost as well. 

To go from the solution of the relaxed mechanism design problem to the actual mechanism design problem where the number of samples purchased is an integer, we will round the allocation obtained in the relaxed problem down to the nearest integer. 

The rounded-down sample purchase scheme remains monotone since it is obtained by applying the monotone floor function to the monotone relaxed allocation rule. Combined with the same Myerson payment rule, the resulting mechanism remains BIC. Due to the rounding, the resulting sample purchasing scheme might not be optimal, but we will show that it is approximately optimal. 

Let $n^* = \bar{\sigma}_{i^*} / \sqrt{\lambda \psi(c_{i^*})}$ be the random variable indicating the number of samples purchased under the relaxation. Recall that the buyer utility is given by 
\begin{equation*}
-\frac{\bar{\sigma}_{i^*}^2}{n^*} - \lambda n^* \psi(c_{i^*}).    
\end{equation*}
Note that rounding down the number of samples $n^*$ increases the second term (as its absolute value decreases). Thus, the only source of suboptimality is that, due to the rounding, the buyer's variance increases. The rounding decreases the number of purchased samples by at most one.  Therefore, the buyer's utility can decrease by at most 
\[\bar{\sigma}_{i^*}^2 \left (\frac 1 {n^*-1} - \frac 1 {n^*} \right ) = \frac{\bar{\sigma}_{i^*}^2}{n^*(n^*-1)}.\]
This is a multiplicative factor of $1/(n^*-1)$ compared to the non-rounded utility. 
\end{proof}

\section{Free-Sample Equilibria}
In the previous section, we showed that the optimal mechanism, under the relaxation that allows the number of samples to be real-valued, is a single-sourcing mechanism, in which the data buyer purchases data from only one data seller. We also showed that an integral approximation of this mechanism yields an approximately optimal mechanism.

In general, finding the optimal mechanism over the space of all non-negative integer allocations is an integer program equivalent to the optimization problem in \eqref{eq:sample_purchasing_problem}, and it may not admit a closed-form solution. For this reason, we work with an approximate equilibrium in this paper.

Alternatively, one could adopt exact equilibrium as the solution concept, but restrict the buyer’s action space so that the mechanism in \Cref{prop:mechanism_soln} is selected. For example, one could assume that it is common knowledge that the buyer employs the approximation approach of solving the relaxed problem and rounding down the resulting sample sizes (the $n_i$’s), or that the buyer commits ex ante to this particular mechanism structure.

Throughout this section, we make the following assumption:
\begin{assumption}[Bounds on the cost density]\label{ass:cost_distribution_bounds}
    The cost density $f$ satisfies $\ell \le f(c) \le L$ for every $c \in [c_{\text{min}}, c_{\text{max}}]$.
\end{assumption}

\subsection{The uninformative equilibrium}
When sellers share a large number of samples at equilibrium, the buyer is more informed about the quality of data that they purchase. On the other hand, when fewer samples are shared at equilibrium, the buyer has a higher degree of uncertainty about the data quality purchased. Our first result shows that for some problem parameters, the equilibrium can be maximally uninformative. That is, there is an equilibrium where no seller shares any free samples. 

\begin{theorem}[Uninformative Equilibrium]\label{thm:uninformative_equilibrium}
Suppose that \Cref{ass:virtual_costs_monotonicity} and \Cref{ass:cost_distribution_bounds} hold and that $\sigma_L > \sqrt{\lambda \psi(c_{\text{max}})}$. Then, for any number of sellers $K$ and any $\zeta>0$, there exist values of $\sigma_L, \sigma_H, \mu$ such that there is a $\zeta$-approximate subgame perfect Nash equilibrium where all sellers share zero free samples. 
\end{theorem}
Note that there are two conflicting forces shaping data sellers’ incentives to share data. On the one hand, competition among sellers may push each seller to reveal their quality; otherwise, they may lose even if they possess high-quality (low-variance) samples. On the other hand, revealing information has two potential drawbacks. First, if a seller’s sample variance is low, they may indeed win, but the buyer will purchase fewer samples from them, as suggested by \Cref{prop:mechanism_soln}, since the buyer can achieve the desired precision with a smaller number of samples when quality is known to be high. Second, if the realized sample variance is high, the seller may regret revealing this information.

Interestingly, as the proof shows, for any number of sellers (that is, for any level of competition) there exist parameter regimes in which the second force dominates, to the point that sellers may prefer not to share any samples for free. In particular, as the proof shows, this outcome arises when the ratio $\sigma_H / \sigma_L$ is sufficiently large, which amplifies the disincentive to share free samples.
\subsubsection{Proof of \Cref{thm:uninformative_equilibrium}}
We claim that the strategy profile in which no seller shares any free samples in the first stage, and the buyer then runs the mechanism in \Cref{prop:mechanism_soln}, constitutes an approximate subgame perfect equilibrium. In the proof, we focus primarily on sellers’ incentives in sharing free samples, and at the end choose parameters so that the mechanism in \Cref{prop:mechanism_soln} achieves the desired approximation factor $\zeta$.

For the sellers, we will first derive a lower bound on the utility before deviation and then construct a distribution over seller variances that results in a strictly lower utility after unilateral deviation. 
    
    \paragraph{Lower bound on utility before deviation.} This lower bound holds for more general symmetric data-sharing strategy profiles where every seller shares the same number of samples, regardless of the number of samples shared and shows a dependence of $\Omega(1/K^2)$ on the number of sellers $K$. In the special case of no samples shared, the lower bound scales with $\sigma_0$, where $\sigma_0 = \sqrt{\mu\sigma_L^2 + (1-\mu)\sigma_H^2}$ is the standard deviation according to the prior.
    
    \begin{lemma}[Utility in symmetric data-sharing strategy profiles]\label{lem:utility_symmetric_strategy}
        Suppose every seller shares the same number of samples $m \in \{0, \ldots, M\}$.  Under the assumption that the pdf of the cost generating distribution $f$ satisfies $\ell \le f(c) \le L$ for every $c \in [c_{\text{min}}, c_{\text{max}}]$, the probability of the buyer buying samples from each seller $i$ is $1/K$. The expected utility of each seller is at least $\Bar{U}_{lb} := \frac{1} {K(K+1)L^2}$. Furthermore, when the number of samples shared $m=0$, the expected utility of each seller is at least $\Bar{U}_{lb}(0) := \sqrt{\frac{ {\mu \sigma_L^2 + (1-\mu) \sigma_H^2}} { {\lambda \psi(c_{\text{max}})}}} \Bar{U}_{lb}$.
    \end{lemma}
    \begin{proof}[Proof of \Cref{lem:utility_symmetric_strategy}]
       By symmetry, each seller has the same probability of winning. So the probability of the buyer buying samples from each seller $i$ is $1/K$. 

       From the Myerson payment from \Cref{prop:mechanism_soln}, the expected utility of each seller is:
       \begin{align*}
           \mathbb{E}_{\bm{c}}[U_i(\bm{c})] = \mathbb{E}_{\bm{c}}[n_i(\bm{c}) (\psi(c_i) - c_i)] = \mathbb{E}_{\bm{c}}[n_i(\bm{c}) {F(c_i)}/{f(c_i)}] \geq \mathbb{E}_{\bm{c}}[n_i(\bm{c}) (c_i - c_\text{min}) \ell / L] 
       \end{align*}
       {Let $N_0$ be a constant that lower bounds the number of samples purchased from a seller if selected. Using the inequality $c_i \ge \min(c_1, \ldots, c_K)$,}
       \begin{align}
           \mathbb{E}_{\bm{c}}[U_i(\bm{c})] &\ge \frac{\ell N_0}{L} \mathbb{E}_{\bm{c}}[\mathds{1}\{\text{buyer buys from } i\} (\min(c_1, \ldots, c_K) - c_{\text{min}})] \label{eqn:expected_utility_lower_bound}
       \end{align}

       From~\Cref{prop:mechanism_soln}, we know that $N_0$ is at least $\bar{\sigma}/\sqrt{\lambda \psi(c_{\text{max}})}$, where $\bar{\sigma}$ is the mean standard deviation according to the posterior. We know in general that $\bar{\sigma} \ge \sigma_L$. For this lower bound on $\bar{\sigma}$, our assumption that $\sigma_L \ge \sqrt{\lambda \psi(c_{\text{max}}))}$ implies that $N_0 \ge 1$. However, when no samples are shared, $\bar{\sigma}$ is always $\sqrt{\mu\sigma_L^2 + (1-\mu)\sigma_H^2} > \sigma_L$. So for $m \ge 1$, we will use $N_0 = 1$ and for $m=0$, we will use $N_0 = \sqrt{\frac{ {\mu \sigma_L^2 + (1-\mu) \sigma_H^2}} { {\lambda \psi(c_{\text{min}})}}}$.
       
       {Consider the lower-bound quantity from \Cref{eqn:expected_utility_lower_bound}. If we sum up the lower bound over all the sellers, we get the value $(\ell N_0/L)\mathbb{E}_{\bm{c}}[\min(c_1, \ldots, c_K)) - c_{\text{max}}]$. This is using the fact that summing $\mathbb{E}_{\bm{c}}[\mathds{1}\{\text{buyer buys from } i\}]$ over all $i \in [K]$ is 1. Additionally, this quantity is equal for every seller by symmetry. Therefore, we can write}
       \begin{align*}
           \mathbb{E}_{\mathbf{c}}[\mathds{1}\{\text{buyer buys from } i\} (\min(c_1, \ldots, c_K)) - c_{\text{min}}] &= \frac{1}{K} \mathbb{E}_{\bm{c}}[(\min(c_1, \ldots, c_K)) - c_{\text{min}}] \\
       \end{align*}

       Finally, we will bound $\mathbb{E}_{\bm{c}}[(\min(c_1, \ldots, c_K)) - c_{\text{min}}]$ to obtain the lower bound on the expected utility. We will again use the lower and upper bounds on cost densities. 
       \begin{align*}
           \mathbb{E}_{\bm{c}}[(\min(c_1, \ldots, c_K)) - c_{\text{min}}] &= \int_{c_\text{min}}^{c_\text{max}} \text{Pr}_{\bm{c}}[\min(c_1, \ldots, c_K) \ge x] dx. \\
           &= \int_{c_\text{min}}^{c_\text{max}} (1 - F(x))^K dx. \\
           &\ge \int_{c_\text{min}}^{c_\text{max}} \frac 1 L(1 - F(x))^K f(x) dx. \\
           &= \frac{1}{L(K+1)}. 
       \end{align*}

       Setting $N_0 = 1$ for $m \ge 1$ and $N_0 = \sqrt{\frac{ {\mu \sigma_L^2 + (1-\mu) \sigma_H^2}} { {\lambda \psi(c_{\text{max}})}}}$ for $m=0$, we get the lower bound on the expected utility in the lemma.
    \end{proof}

    \paragraph{Upper bound on utility after deviation.}
    We show that after deviating, the buyer's posterior belief that the seller's variance is $\highvar$ shifts: it decreases from $1 - \mu$ to below $1 - \mu - \Delta_L$ when $\lowvar$ is realized, and increases to above $1 - \mu + \Delta_H$ when $\highvar$ is realized. Here $\Delta_L, \Delta_H > 0$ are constants to be chosen later. We show that this shift occurs with high probability over the samples, and that this probability approaches one as $\sigma_H / \sigma_L \to \infty$.  
    
    This belief shift reduces the deviating seller's utility for both realized variance types, though for different reasons. When $\lowvar$ is realized, the shift toward $\lowvar$ reduces the utility upon selection (payment minus cost). When $\highvar$ is realized, the shift toward $\highvar$ reduces the probability of being selected.

    \begin{lemma}[Belief shift from sharing $m$ samples]\label{lem:upper_bound_belief_shift}
        For any $\Delta_L \in (0,\, 1-\mu)$, $\Delta_H \in (0,\, \mu)$, and $\delta \in (0,1)$, there exists $R_0 > 0$ depending on $\mu, \Delta_L, \Delta_H, \delta$, such that whenever $\highsigma/\lowsigma > R_0$ and any $m \ge 2$ free samples are shared, the posterior $\pi_H := \Pr(\sigma^2 = \highvar \mid S^2)$ satisfies:
        \begin{enumerate}[label=(\roman*)]
            \item $\Pr\bigl(\pi_H \le 1 - \mu - \Delta_L \;\big|\; \sigma^2 = \lowvar\bigr) \ge 1 - \delta$, and
            \item $\Pr\bigl(\pi_H \ge 1 - \mu + \Delta_H \;\big|\; \sigma^2 = \highvar\bigr) \ge 1 - \delta$.
        \end{enumerate}
    \end{lemma}

    \begin{proof}[Proof of \Cref{lem:upper_bound_belief_shift}]
        Let $S^2$ be the sample variance of the free samples. The Bayesian belief update based on $S^2$ results in a posterior belief with probability of variance is $\highvar$ given by \[\pi_H = \Pr(\sigma^2 = \highvar | S^2) = \frac{1}{1 + \frac{\mu}{1-\mu}\Lambda_m(S^2)},\]
        where $\Lambda_m(S^2) = f_{S^2|\lowsigma}(S^2) / f_{S^2|\highsigma}(S^2)$ is the likelihood ratio based of the sample variance of $m$ samples which can be written as 
        \[\Lambda_m = \left(\frac{\highvar}{\lowvar}\right)^{(m-1)/2} \exp\left(-\frac{\highvar - \lowvar}{2\highvar \lowvar}\, (m-1) S^2\right).\]

        To obtain $\pi_H < 1 - \mu - \Delta$, the likelihood ratio must satisfy
        \[\Lambda_m > T_+ := \frac{(1-\mu)(\mu + \Delta)}{\mu(1 - \mu - \Delta)}.\]
        Under $\lowvar$, let $Z = (m-1)S^2 / \lowvar \sim \chi^2_{m-1}$. Then
        \[\log \Lambda_m = (m-1)\log\frac{\highsigma}{\lowsigma} - \frac{\highvar - \lowvar}{2\highvar}\, Z.\]
        The condition $\Lambda_m > T_+$ is equivalent to $Z < z_+$, where
        \[z_+ := \frac{2\highvar}{\highvar - \lowvar}\left((m-1)\log\frac{\highsigma}{\lowsigma} - \log T_+\right).\]
        Note that $z_+$ is well defined when $T_+ > 0$, which is true when $\Delta \in (0, 1-\mu)$. Additionally, $z_+$ is positive when $\highsigma / \lowsigma$ is chosen to be larger than $T_+$. In this case, as $\highsigma / \lowsigma \to \infty$, $z_+ \to \infty$. As a result, $\Pr(Z \le z_+) \to 1$. Therefore, there is a choice of $\highsigma / \lowsigma$ large enough that $\Pr(Z \le z_+) \ge 1 - \delta$.

        To obtain $\pi_H > 1 - \mu + \Delta$, the likelihood ratio must satisfy
        \[\Lambda_m < T_- := \frac{(1-\mu)(\mu - \Delta)}{\mu(1 - \mu + \Delta)}.\]
        Under $\highvar$, let $Z' = (m-1)S^2/\highvar \sim \chi^2_{m-1}$. Then
        \[\log \Lambda_m = (m-1)\log\frac{\highsigma}{\lowsigma} - \frac{\highvar - \lowvar}{2\lowvar}\, Z'.\]
        The condition $\Lambda_m < T_-$ is equivalent to $Z' > z_-$, where
        \[z_- := \frac{2\lowvar}{\highvar - \lowvar}\left((m-1)\log\frac{\highsigma}{\lowsigma} - \log T_-\right).\]
        Note that $z_-$ is well defined when $T_- > 0$, which holds when $\Delta \in (0, \mu)$. As $\highsigma / \lowsigma \to \infty$, so $z_- \to 0$. As a result, $\Pr(Z' \ge z_-) \to 1$. Therefore, there is a choice of $\highsigma / \lowsigma$ large enough that $\Pr(Z' \ge z_-) \ge 1 - \delta$.

        Setting $R_0$ large enough that both tail probabilities are at most $\delta$ completes the proof.
    \end{proof}

    \paragraph{Upper bound on utility for low variance after deviation.} 

    Fix $\delta \in (0,1)$, $\Delta_L \in (0, 1 - \mu)$, and $\Delta_H \in (0, \mu)$. We will state the specific values of these parameters later. Suppose $\sigma_H/\sigma_L$ is large enough that, by \Cref{lem:upper_bound_belief_shift}, the posterior belief of $\highvar$ is at most $1 - \mu - \Delta_L$ when $\lowvar$ is realized and at least $1 - \mu + \Delta_H$ when $\highvar$ is realized, each with probability at least $1 - \delta$. 
    
    In the complementary event (probability at most $\delta$), we use the trivial upper bounds: selection probability $1$ and maximum payment $\sigma_H \sqrt{\psi(c_{\max}) / \lambda}$.
    
    When $\lowvar$ is realized and the belief shift holds, we bound the selection probability by $1$. The payment upon selection is at most  $\sqrt{\mathbb{E}[\sigma^2 \mid \lowvar]\,\psi(c_{\max}) / \lambda}$, which is at most \\ $\sqrt{((\mu + \Delta_L)\sigma_L^2 + (1 - \mu - \Delta_L) \sigma_H^2)\, \psi(c_{\max})/\lambda}$.
    
    When $\highvar$ is realized and the belief shift holds, we bound the utility upon selection by $\sigma_H \sqrt{\psi(c_{\max}) / \lambda}$. It remains to bound the selection probability. The expected posterior variance is at least $\sigma'^{\,2}(\Delta_H) := (\mu - \Delta_H)\lowvar + (1 - \mu + \Delta_H)\highvar$, which is strictly greater than $\sigma_0^2 = \mu\lowvar + (1-\mu)\highvar$. 
    Whenever $\sigma'^{\,2}(\Delta_H)\, \psi(c_{\min}) > \sigma_0^2\, \psi(c_{\max})$, the deviating seller is never selected when $\highvar$ is realized.
    
    This condition is equivalent to $\sigma'^{\,2}(\Delta_H)/\sigma_0^2 > \rho$, where $\rho := \psi(c_{\max})/\psi(c_{\min})$. Rearranging:
    \begin{align*}
        \frac{\sigma_0^2 + \Delta_H(\highvar - \lowvar)}{\sigma_0^2} > \rho 
        \quad\Longleftrightarrow\quad
        \frac{\highvar}{\lowvar} > \frac{1 + \mu(\rho - 1)}{\Delta_H - (1-\mu)(\rho - 1)}.
    \end{align*}
    Setting $\Delta_H = \mu/2$ and choosing $\mu > \frac{2\rho}{1 + 2\rho}$ ensures the denominator is positive. Then for $\highvar/\lowvar$ sufficiently large, the deviating seller is never selected when $\highvar$ is realized.

    Putting the pieces together, the expected utility after deviation is at most
    \begin{align*}
        \mathbb{E}_{\bm{c}}[U_i(\bm{c})] &\le \left(\mu \sqrt{(\mu + \Delta_L)\lowvar +  (1 - \mu - \Delta_L)\highvar} + \delta\, \highsigma \right) \sqrt{\frac{\psi(c_{\max})} {\lambda}}\,.
    \end{align*}

    We have fixed $\mu > \frac{2\rho}{1+2\rho}$ and $\Delta_H = \mu/2$. It remains to choose $\Delta_L$ and $\delta$ so that the upper bound above is strictly less than $\Bar{U}_{\mathrm{lb}}(0)$. This requires
    \begin{align*}
        \mu \sqrt{(\mu + \Delta_L)\lowvar +  (1 - \mu - \Delta_L)\highvar} + \delta\, \highsigma &\le \sqrt{\frac{(\mu \sigma_L^2 + (1-\mu) \sigma_H^2)\,\ell\,(c_{\max} - c_{\min})}{K(K+1)L^2}}\,.
    \end{align*}
    Setting $\delta = 1 - \mu - \Delta_L$, this reduces to showing
    \[
    \frac{(\mu + \Delta_L)\lowvar +  (1 - \mu - \Delta_L)\highvar}{\mu \lowvar + (1-\mu)\highvar} \le \xi,
    \]
    where $\xi > 0$ is a constant depending on $c_{\min}, c_{\max}, \ell, L, K$. As $\Delta_L \to 1 - \mu$ and $\highvar/\lowvar \to \infty$, the left-hand side tends to zero. Hence there exists $\epsilon > 0$ and $R_1 > 0$ such that setting $\Delta_L = 1 - \mu - \epsilon$ and $\delta = \epsilon$ satisfies the inequality whenever $\highsigma / \lowsigma > R_1$. 
    
    Finally, we choose $\sigma_L$ large enough to satisfy \eqref{eqn:approx_condition} with $\alpha$ such that $1/(\alpha-1) < \zeta$, and also choose $\sigma_H > \sigma_L.\max(R_0, {R_1})$, where $R_0$ is from \Cref{lem:upper_bound_belief_shift} to ensure belief shift holds with the chosen $\Delta_H, \Delta_L, \delta$. This ensures all conditions hold simultaneously, completing the proof.
\subsection{The maximally informative equilibrium}

In this section, we ask the question the other way around. If we fix the problem’s parameters and increase the number of sellers, does the incentive to share data dominate? It turns out that the answer is yes. In fact, as the number of sellers grows, sharing the maximum possible number of samples $M$ becomes an equilibrium---and indeed the unique equilibrium. The following result formalizes this claim.

\begin{theorem}[Equilibrium is maximally informative when there are many sellers] \label{theorem:maximally_informative}
Suppose that \Cref{ass:virtual_costs_monotonicity} and \Cref{ass:cost_distribution_bounds} hold.  Let $\sigma_L \ge \alpha \sqrt{\lambda \psi(c_{\max})}$ for some $\alpha \ge 3$. Then:
\begin{itemize}
\item For every $\sigma_L, \sigma_H, \mu$, there exists $\bar{K}$, such that, for every number of sellers $K \geq \bar{K}$, the strategy profile in which each seller shares the maximum possible number of free samples ($M$) and the buyer runs the mechanism in \Cref{prop:mechanism_soln} is a $1/(\alpha-1)$-approximate subgame perfect equilibrium.
\item Assuming the buyer runs the mechanism in \Cref{prop:mechanism_soln}, the above strategy profile is the unique equilibrium.
\end{itemize}
\end{theorem}
\subsubsection{Proof of \Cref{theorem:maximally_informative}}

    The key driver of the incentive to sharing the maximum number of samples is that it allows the mean variance under the posterior to be low enough which makes the seller more prone to being selected. There is a fundamental limit to how low the posterior mean can be if a seller shares fewer than $M$ samples. This is given in the following lemma which establishes a lower bound on posterior mean as a function of number of free samples shared $(m)$.

     \begin{lemma}[Lower bound on mean of posterior]\label{lem:lower_bound_posterior}
        For any free sample set $S^f_m$ of size $m \ge 2$, the mean of the posterior over variance (based on the sample variance $S^2$ of the free samples) is at least
        \[\sigma_{lb}^2(m) := \sigma_L^2 + \frac{\sigma_H^2 - \sigma_L^2}{1 + \frac{\mu}{1-\mu} \left (\frac{\sigma_H}{\sigma_L} \right )^{m-1}}.\]
        Moreover for any $\sigma' > \sigma_{lb}$, there is a positive probability that the mean of the posterior is $\ge \sigma'$.
     \end{lemma}

\begin{proof}[Proof of \Cref{lem:lower_bound_posterior}]
    The mean of the posterior distribution over variance can be written as
    \begin{align*}
        \mathbb{E}[\sigma^2 | S^2] &= \lowvar (1 - \pi_H) + \highvar \pi_H \\
        &= \lowvar + \pi_H (\highvar - \lowvar),
    \end{align*}
    where $\pi_H$ is the probability of variance $\highvar$ under the posterior based on the sample variance $S^2$ of $S^f_m = \{X_1, \ldots, X_m\}$.
    By Bayes' theorem,
    \begin{align*}
        \pi_H = \Pr(\sigma^2 = \highvar | S^2) = \frac{1}{1 + \frac{\mu}{1-\mu}\Lambda_m(S^2)},
    \end{align*}
    where $\Lambda_m(S^2) = f_{S^2|\lowsigma}(S^2) / f_{S^2|\highsigma}(S^2)$ is the likelihood ratio based on the sample variance of $m$ samples. Each likelihood function is the density of a scaled chi-squared distribution with $m-1$ degrees of freedom. Hence, the likelihood ratio can be written as: 
    \[\Lambda_m = \left(\frac{\highvar}{\lowvar}\right)^{(m-1)/2} \exp\left(-\frac{\highvar - \lowvar}{2\highvar \lowvar}\, (m-1)S^2\right).\]
    Since $S^2 \ge 0$ and $\highvar > \lowvar$, the exponential term is non-positive, so
    \[\Lambda_m \leq \left(\frac{\highvar}{\lowvar}\right)^{(m-1)/2} = \left(\frac{\highsigma}{\lowsigma}\right)^{m-1}.\]
    This upper bound on $\Lambda_m$ gives a lower bound on $\pi_H$:
    \[\pi_H \geq \frac{1}{1 + \frac{\mu}{1-\mu}\left(\frac{\highsigma}{\lowsigma}\right)^{m-1}},\]
    which yields the lower bound $\sigma_{lb}^2(m)$ stated in the lemma.

    For the moreover claim, fix any $\sigma_{lb} < \sigma' < \sigma_H$. There is a corresponding threshold $\Lambda' < (\highsigma/\lowsigma)^{m-1}$ such that the posterior mean equals $\sigma'^2$ when $\Lambda_m = \Lambda'$. This corresponds to $S^2$ exceeding a strictly positive value, which has positive probability.
\end{proof}

\newcommand{\esel}{E^{\text{select}}}
\newcommand{\ebar}{\bar{E}^{\text{select}}}
For any seller $i$ consider the event 
\[\esel_i = \{c_i \le c_0, \bar{\sigma_{i}} \le \sigma_0\},\]
where $c_0, \sigma_0$ are defined in the following way. Let $\sigma_0 = (\sigma_{lb}(M-1) + \sigma_{lb}(M))/2$, where $\sigma_{lb}$ is the lower bound on posterior mean given by \Cref{lem:lower_bound_posterior}. And $c_0 \gneq c_{\text{min}}$ satisfies the condition 
\[\sigma_0 \sqrt{\psi(c_0)} \le \sigma_{lb}(M-1) \sqrt{\psi(c_{\text{min}})}.\]
Such a $c_0$ exists since $\sigma_{lb}(M-1)$ is strictly bigger than $\sigma_0$ (given that $\sigma_{lb}(M-1) > \sigma_{lb}(M)$ as shown in \Cref{lem:lower_bound_posterior}) and since we assume the virtual cost function $\psi$ to be continuous and increasing.

\begin{lemma}\label{lem:win_suff_cond}
    Conditioned on the event $\esel$ for some seller $i$, no seller $i'$ sharing $m < M$ free samples will be selected. For any seller $i$ sharing $M$ free samples, there is a positive probability of $\esel_i$.
\end{lemma}
\begin{proof}[Proof of \Cref{lem:win_suff_cond}]
    For any seller $i'$ sharing $m < M$ samples, the posterior mean satisfies $\bar{\sigma}_i \ge \sigma_{lb}(m)$ by \Cref{lem:lower_bound_posterior}. From that lemma, we also see that $\sigma_{lb}(m) \ge \sigma_{lb}(M-1) > \sigma_{lb}(M)$ since the lower bound established there is strictly decreasing as number of free samples increases. Hence $\bar{\sigma}_{i'}$ is strictly greater than $\sigma_0$ which is midpoint of $\sigma_{lb}(M-1)$ and $\sigma_{lb}(M)$. Additionally, $c_{i'} \ge c_{\text{min}}$. As a result, $\sigma_0 \sqrt{\psi(c_0)} \le \sigma_{lb}(M-1) \sqrt{\psi(c_{\text{min}})} < \bar{\sigma}_{i'} \sqrt{\psi(c_{i'})} $.

    If there is a seller $i$ with event $\esel_i$ realized, then for this seller, $\bar{\sigma}_i \sqrt{\psi(c_i)} \le \sigma_0 \sqrt{\psi(c_0)}$. Comparing to seller $i'$ sharing $m < M$ samples, $\bar{\sigma}_{i} \sqrt{\psi(c_{i})} < \bar{\sigma}_{i'} \sqrt{\psi(c_{i'})}$ meaning that seller $i'$ will not be selected.

    Any seller $i$ sharing $M$ samples has a positive probability of having $\bar{\sigma}_{i} \le \sigma_0$. This is due to $\sigma_0$ being strictly less than $\sigma_{lb}(M)$. \Cref{lem:lower_bound_posterior} establishes a positive probability for this. Furthermore, there is a positive probability of $c_i \in [c_0/2, c_0]$ which is at least $\ell c_0 / 2(c_{\text{max}} - c_{\text{min}})$ from our assumption on the lower bound on densities of the cost distribution. Therefore, there is a positive probability of $\esel_i$. 
\end{proof}

A consequence of this lemma is that the probability of a seller sharing $m < M$ samples being selected decays exponentially in the number of sellers sharing $M$ samples. This is because for the seller to be selected, the event $\esel$ must not be realized for \textit{any} of the $M$ samples sharing sellers, and this event is independent for each of the sellers. This is stated in the following lemma.

\begin{lemma}[Upper bound on expected utility when sharing $m < M$ samples]\label{lem:win_prob_low_samples}
    Suppose $J$ of the $K$ sellers share $M$ samples. Then the expected utility of a seller sharing $m < M$ samples is at most $(1 - q)^J \sigma_H \sqrt{\psi(c_{\max}) / \lambda}$, where for $q > 0$ is the probability of event $\esel$ for a seller sharing $M$ samples.
\end{lemma}
\begin{proof}[Proof of \Cref{lem:win_prob_low_samples}]
We will show that the probability of a seller sharing $m < M$ samples being selected is at most $(1 - q)^J$ and this results in the upper bound on utility assuming maximum utility upon selection. This upper bound on selection probability follows directly from \Cref{lem:win_suff_cond}.  From \Cref{lem:win_suff_cond}, we know that if there is any seller $i$ sharing $M$ samples, there is a positive probability $\ge q$ of event $\esel_i$ which results in the seller $i'$ sharing $m < M$ samples not being selected. So for seller $i'$ to be selected, each of the $J$ sellers sharing $M$ samples should not have event $\esel$ realized. This event is independent across all sellers. So the probability of this event not occurring for any seller is $(1 - q)^J$.
\end{proof}

Next we will show a lower bound on utility of a seller sharing $M$ samples. This lemma shows that the expected utility of a seller sharing $M$ samples is of the order $\Omega(1/J^2)$, where $J$ is the number of sellers sharing $M$ samples.

\begin{lemma}[Lower bound on utility when sharing $M$ samples]\label{lem:win_prob_high_samples}
    Suppose $J$ of the $K$ sellers share $M$ samples. Then the expected utility of a seller sharing $M$ samples is at least 
    \[\frac {(1 - (1-q)^J)\ell}{L^2 J^2} \frac{\frac 1 J - (1 - \ell(c_0 - c_{\text{min}}))^J}{(1 - (1 - \ell(c_0 - c_{\text{min}}))^J)} \in \Omega(1/J^3),\] 
    where $q > 0$ is the probability of event $\esel$ for a seller sharing $M$ samples.
\end{lemma}
\begin{proof}[Proof of \Cref{lem:win_prob_high_samples}]
    Let $\mathcal{J}$ denote the index set of the $J$ sellers sharing $M$ samples. Consider the event $\bar{E}^{\text{select}} = \bigcup_{j \in \mathcal{J}} \esel_i$ which is the event that one of the $J$ sellers sharing $M$ samples realizes event $\esel$. Whenever $\bar{E}^{\text{select}}$ occurs, one of the sellers in $\mathcal{J}$ gets selected by \Cref{lem:win_suff_cond}. The probability of $\bar{E}^{\text{select}}$ is at least $1 - (1-q)^J$. We provide a lower bound on the expected utility by a lower bound on the expected utility conditioned on $\bar{E}^{\text{select}}$ times the probability of $\bar{E}^{\text{select}}$.

    To bound the expected utility conditioned on $\bar{E}^{\text{select}}$, we will perform an analysis similar to the analysis in \Cref{lem:utility_symmetric_strategy} which expresses a lower bound on expected utility of a seller as in \Cref{eqn:expected_utility_lower_bound}. We will extend this analysis to bound the conditional expected utility. 

    Note that conditioning on $\bar{E}^{\text{select}}$ retains the symmetry in distributions of sellers in $\mathcal{J}$. Note that the event $\bar{E}^{\text{select}}$ is symmetric for $\mathcal{J}$.  That is, if a tuple $((c_j, \bar{\sigma}_j))_{j \in \mathcal{J}} \in \bar{E}^{\text{select}}$, then $((c_{\beta(j)}, \bar{\sigma}_{\beta(j)})_{j \in \mathcal{J}} \in \bar{E}^{\text{select}}$ for every permutation $\beta$ of $\mathcal{J}$. As a result the distribution of each $(c_j, \bar{\sigma}_j)$ conditioned on event $\bar{E}^{\text{select}}$ is the same for every $j \in \mathcal{J}$. 

    As we did in \Cref{lem:utility_symmetric_strategy}, using the payment derived in \Cref{prop:mechanism_soln}, we can write a lower bound on expected utility of a seller in $J$ conditioned on the event $\ebar$ as
    \begin{align}
           \mathbb{E}_{\bm{c}}[U_j(\bm{c}) | \ebar] &\ge \frac{\ell}{L} \mathbb{E}_{\bm{c}}[\mathds{1}\{\text{buyer buys from } j\} (\min_{j \in \mathcal{J}} c_j - c_{\text{min}}) | \ebar]. \label{eqn:expected_cond_utility_lower_bound}
       \end{align}

       {Consider the lower-bound quantity from \Cref{eqn:expected_cond_utility_lower_bound}. If we sum up the lower bound over all the sellers, we get the value $\ell/L\mathbb{E}_{\bm{c}}[\min_{j \in \mathcal{J}} c_j - c_{\text{min}} | \ebar]$. This uses the fact that summing $\mathbb{E}_{\bm{c}}[\mathds{1}\{\text{buyer buys from } j | \ebar\}]$ over all $j \in \mathcal{J}$ is one. Additionally, this quantity is equal for every seller by symmetry since the conditional distribution of every seller's parameter, conditioned on $\ebar$ is the same. Therefore, we can write}
       \begin{align*}
           \mathbb{E}_{\bm{c}}[U_j(\bm{c}) | \ebar] &\ge \frac{\ell }{LJ} \mathbb{E}_{\bm{c}}[(\min_{j \in \mathcal{J}} c_j - c_{\text{min}} | \ebar] \\
       \end{align*}

       Finally, we  bound $\mathbb{E}_{\bm{c}}[(\min_{j \in \mathcal{J}} c_j - c_{\text{min}} | \ebar]$ to obtain the lower bound on the expected utility. We again use the lower and upper bounds on cost densities: 
        \begin{align*}
            \mathbb{E}_{\bm{c}}[\min_{j \in \mathcal{J}} c_j - c_{\text{min}} | \ebar] &= \mathbb{E}_{\bm{c}}[\min_{j \in \mathcal{J}} c_j - c_{\text{min}} | \exists j \in \mathcal{J} \text{ s.t. } c_j < c_0, \bar{\sigma_j} < \sigma_0].
            \intertext{By independence of costs and sample distributions and the fact that the cost and sample distributions conditioned on $\bar{E}^{\mathrm{select}}$ is the same for every seller in $\mathcal{J}$, the probability of a seller having cost lower than $c_0$ having posterior mean lower than $\sigma_0^2$ is $1/|J|$. Consequently, }
            \mathbb{E}_{\bm{c}}[\min_{j \in \mathcal{J}} c_j - c_{\text{min}} | \exists j \in \mathcal{J} \text{ s.t. } c_j < c_0, \bar{\sigma_j} < \sigma_0] & \ge \frac 1 J \ \mathbb{E}_{\bm{c}}[\min_{j \in \mathcal{J}} c_j - c_{\text{min}} | \exists j \in \mathcal{J} \text{ s.t. } c_j < c_0] \\
            &= \frac 1 J \int_{c_{\text{min}}}^{c_{\text{max}}} \Pr[\min_{j \in \mathcal{J}} c_j \ge x | \min_{j \in \mathcal{J}} c_j \le c_0] dx  \\
            &= \frac 1 J \int_{c_{\text{min}}}^{c_0} \frac{\Pr[\min_{j \in \mathcal{J}} c_j \ge x] - \Pr[\min_{j \in \mathcal{J}} c_j \ge c_0]}{\Pr[\min_{j \in \mathcal{J}} c_j \le c_0]} dx  \\
            &= \frac 1 J \int_{c_{\text{min}}}^{c_0} \frac{(1 - F(x))^J - (1 - F(c_0))^J}{1 - (1 - F(c_0))^J} dx 
            \intertext{Using, $f(x) \le L$ for all $x \in [c_{\text{min}}, c_{\text{max}}]$,}
            \mathbb{E}_{\bm{c}}[\min_{j \in \mathcal{J}} c_j - c_{\text{min}} | \exists j \in \mathcal{J} \text{ s.t. } c_j < c_0, \bar{\sigma_j} < \sigma_0]  &  \ge \frac 1 J \frac{\frac 1 J - \left ( 1 - \ell(c_0 - c_{\text{min}})\right )^J}{L \left (1 - \left ( 1 - L(c_0 - c_{\text{min}})\right )^J \right ).}
        \end{align*}
        
       Finally, combining this lower bound on expected utility conditioned on $\ebar$ and the probability of $(1 - (1-q)^J)$ of event $\ebar$ concludes the proof. 
    
\end{proof}       

Using these lemmas, we now establish that the unique approximate equilibrium is one in which all sellers share the maximum number of samples.

\paragraph{Deviation from maximally informative strategy profile.} Consider a unilateral deviation from the profile with all sellers sharing $M$ samples to sharing $m$ samples. From \Cref{lem:utility_symmetric_strategy}, we know that before deviation, the utility of each seller is at least $\Omega(1/K^2)$. After deviation, from \Cref{lem:win_prob_low_samples}, we know the expected utility is at most $O((1 - q)^{K-1})$ for a $q > 0$ corresponding to the probability of event $\esel$ for a non-deviating seller. Since deviation results in expected utility going from decaying quadratically in number of sellers $K$ to decaying exponentially, for sufficiently large $K$, the utility after deviation is lower than the utility before deviation, making the maximally informative sample sharing strategy an equilibrium. 

\paragraph{Deviation from any other sample sharing profile.} Consider any other strategy profile where at least one seller shares strictly fewer than $M$ samples. Consider the seller $i'$ sharing $m < M$ samples with the least probability of being selected. We will show if $i'$ unilaterally deviates to sharing $M$ samples, the expected utility of $i'$ strictly increases. 

Let us consider two cases for $J$, which is the number of sellers sharing $M$ samples. The first case is that $J \ge K^{1/4}$. In this case, by \Cref{lem:win_prob_low_samples}, seller $i$'s expected utility before deviation is $O((1-q)^{K^{1/4}})$. After deviation, since $i'$ becomes one among the $J+1$ sellers sharing $M$ samples, the expected utility becomes at least $\Omega(1/K^3)$ (since $J \le K$). So far large enough $K$, deviation strictly increases utility.

The other case is $J \le K^{1/4}$. Since $i'$ is the one among the $K - J$ sellers with the least probability of being selected, $i'$ has a probability of being selected, and hence expected utility, at most $O(1/(K - \sqrt{K}))$. After deviation, the expected utility becomes $\Omega(1/K^{3/4})$. Again for large enough $K$, deviation strictly increases utility. 

Choosing a large enough $K$, larger than what is required for both cases, concludes the proof.

\section{Numerical Simulations}\label{sec:experiments}

Theorems~\ref{thm:uninformative_equilibrium} and~\ref{theorem:maximally_informative} characterize the uninformative and maximally informative equilibria in limiting regimes: the former when $\highsigma/\lowsigma$ is large, the latter when $K$ is large. To map out the equilibrium structure across the full parameter space---and in particular to probe whether additional equilibria exist---we conduct Monte Carlo simulations.

We simulate the game with $K$ symmetric sellers whose per-sample costs are drawn i.i.d.\ from $\mathrm{Uniform}[c_{\min}, c_{\max}]$ and whose data variances are $\lowvar$ with probability $\mu$ and $\highvar$ with probability $1-\mu$. The buyer runs the approximately optimal mechanism of \Cref{prop:mechanism_soln}. We fix $\highsigma = 50$ and vary $\lowsigma = \highsigma / r$ to sweep the ratio $r = \highsigma/\lowsigma$, with $\mu = 0.6$, $\lambda = 0.007$, $c_{\min} = 0.5$, $c_{\max} = 2.0$, and $M = 5$. For each candidate symmetric strategy $m^* \in \{0, 2,3, \dots, M\}$, we estimate seller~1's expected profit under every deviation using $N = 100{,}000$ Monte Carlo samples. We consider symmetric strategies and find the strategy where each seller shares $m^*$ samples an equilibrium only if its estimated profit exceeds every alternative by at least two combined standard errors. While this approach ensures that the detected equilibria are indeed equilibria with significant certainty, it cannot detect equilibria where deviation does not decrease utility or decreases utility by a small amount. 

\Cref{fig:phase_diagram} maps the equilibrium structure across $K \in \{2, \dots, 10\}$ and $\highsigma/\lowsigma \in \{2, \dots, 300\}$. Five regions emerge:

\begin{figure}[t]
    \centering
    \includegraphics[width=\textwidth]{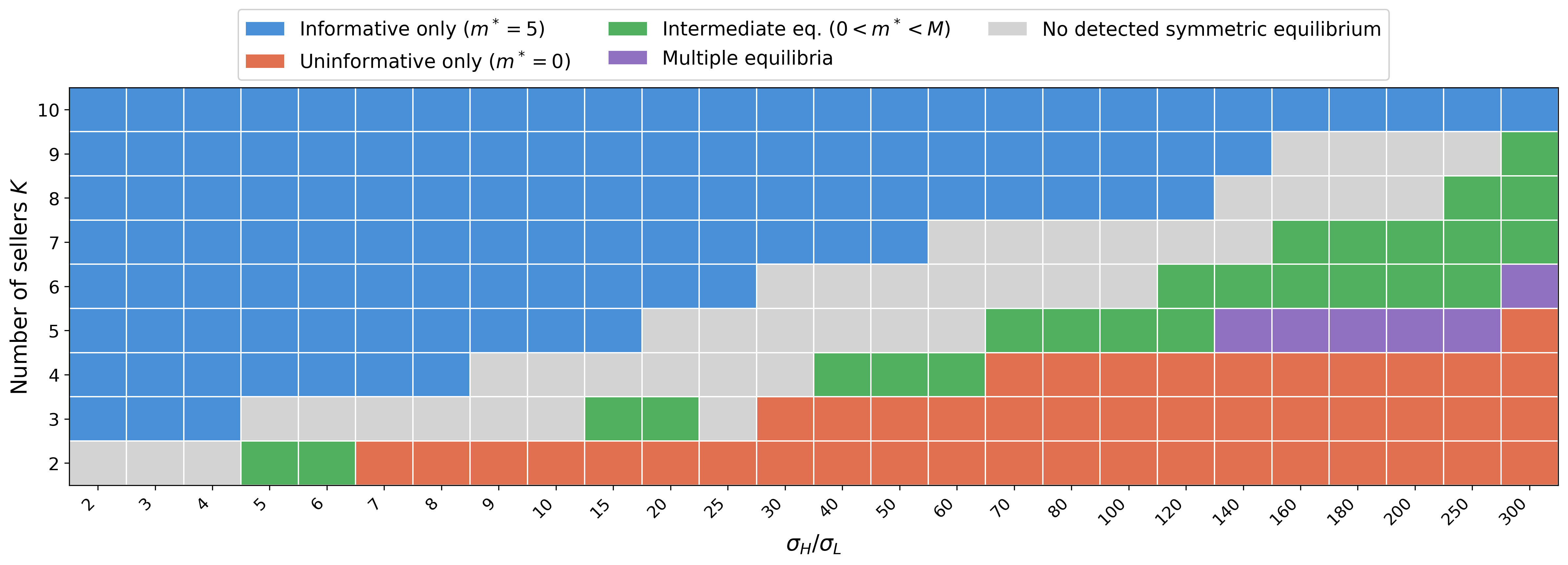}
    \caption{Phase diagram of symmetric equilibria across $(K,\, \highsigma/\lowsigma)$.  Each cell is colored based on the equilibrium detected in the corresponding regime of parameters. Blue: only the informative equilibrium ($m^* = M$). Red: only the uninformative ($m^* = 0$). Green: only an intermediate equilibrium ($m^* = 2$). Purple: coexistence of multiple symmetric equilibria. Gray: no symmetric equilibrium found.}
    \label{fig:phase_diagram}
\end{figure}

\begin{enumerate}[leftmargin=*]
    \item \emph{Informative equilibrium} (blue). At low $\highsigma/\lowsigma$ or high $K$, the unique symmetric equilibrium is $m^* = M$. For $K \ge 10$, this is the only equilibrium at every ratio tested.

    \item \emph{Uninformative equilibrium} (red). At high $\highsigma/\lowsigma$ and small $K$, the equilibrium $m^* = 0$ exists. 

    \item \emph{Intermediate equilibrium} (green). Between the maximally informative and maximally uninformative regions, we have an intermediate equilibrium with $m^* = 2$ emerges. This equilibrium appears for every $K$ from $2$ to $9$ at appropriate ratios.

    \item \emph{Multiple equilibria} (purple). For a band of parameters (e.g., $K=5$ at $\highsigma/\lowsigma \in [140, 250]$; $K=6$ at $\highsigma/\lowsigma =300 $), both $m^* = 0$ and $m^* = 2$ are simultaneously equilibria. 

    \item \emph{No detected symmetric equilibrium} (gray). In some regions, we do not detect any symmetric strategies to be equilibria. Note that this does not mean that there are no symmetric equilibria for those parameters. But rather, we are not able to detect it given the uncertainty in our Monte Carlo sampling. 
\end{enumerate}

\section{Conclusion}
This paper studies the strategic role of free trials in competitive data markets. We model free trials as a pre-auction disclosure decision: before a buyer arrives, each seller commits to releasing a limited number of free samples, after which the buyer runs a procurement mechanism to purchase data. To make the continuation game tractable, we analyze a simple mechanism that (approximately) optimizes the buyer’s tradeoff between statistical accuracy and procurement cost. The resulting continuation policy induces a clear, belief-adjusted comparison of sellers based on both reported costs and posterior beliefs about data quality inferred from the free samples.

Our analysis delivers two main insights. First, free trials need not emerge in equilibrium: when the quality gap is sufficiently large, sellers may prefer to remain uninformative rather than risk revealing adverse information. Second, when competition is sufficiently intense, maximal disclosure becomes the unique prediction: with many sellers, the probability that at least one seller appears highly attractive under the buyer’s selection rule increases, and free samples become a decisive signal for winning the procurement auction. 

There are several natural directions for future work. One extension is to allow sellers’ datasets to be heterogeneous not only in variance but also in mean, so that sellers’ data may exhibit different biases or systematic shifts relative to the buyer’s target. A second direction is to relax the assumption that free samples are i.i.d. draws from the same distribution as the purchased data. 
Finally, it would be valuable to extend the analysis beyond Gaussian data. In this paper the Gaussian assumption is used primarily to obtain a tractable likelihood for the sample variance and to derive likelihood-ratio bounds that drive the belief-updating and equilibrium arguments. Similar results should be attainable for broader families of distributions as long as one can derive suitable bounds on the relevant likelihood ratios (or comparable concentration bounds) for the statistic used to summarize free samples.

\section{Acknowledgments}
Part of this work was supported by the Simons Institute for the Theory of Computing, and conducted when Alireza Fallah visited the Institute.
We also wish to acknowledge funding by the European Union (ERC-2022-SYG-OCEAN-101071601).
Views and opinions expressed are however those of the author(s) only and do not
necessarily reflect those of the European Union or the European Research Council
Executive Agency. Neither the European Union nor the granting authority can be
held responsible for them.
\bibliography{references}

\end{document}